\documentclass[12pt,a4paper]{amsart}  
\usepackage{cmbright}
\usepackage[T1]{fontenc}

\usepackage{amsmath,amsthm,amssymb,enumerate,amsfonts,graphicx}
\usepackage{color}
\usepackage{epsfig}
\usepackage[numbers,sort&compress]{natbib}
\usepackage[T1]{fontenc}
\usepackage[utf8]{inputenc}
\usepackage{algorithm,algorithmic}
\usepackage{stmaryrd}
\usepackage{paralist}
\usepackage{tikz}
\tikzset{black node/.style={draw, circle, fill = black, minimum size = 5pt, inner sep = 0pt}}
\tikzset{normal/.style = {draw=none, fill = none, minimum size =0, rectangle}}

\usepackage{hyperref}

\hypersetup{
  pdftitle = {A tight Erdös-Pósa function for wheel minors},
  pdfauthor = {Aboulker, Fiorini, Huynh, Joret, Raymond, Sau},
  colorlinks  
}

\usepackage{aliascnt}

\newtheorem{theorem}{Theorem}[section]

\newaliascnt{lemma}{theorem}
\newtheorem{lemma}[lemma]{Lemma}
\aliascntresetthe{lemma}

\newaliascnt{corollary}{theorem}
\newtheorem{corollary}[corollary]{Corollary}
\aliascntresetthe{corollary}

\newaliascnt{conjecture}{theorem}
\newtheorem{conjecture}[conjecture]{Conjecture}
\aliascntresetthe{conjecture}

\newaliascnt{claim}{theorem}
\newtheorem{claim}[claim]{Claim}
\aliascntresetthe{claim}

\newcommand{\R}{\mathbb{R}}

\newcommand{\N}{\mathbb{N}}

\DeclareMathOperator{\tw}{\mathsf{tw}}

\newcommand{\mc}{\mathcal}

\newcommand{\EP}{Erd\H{o}s-P\'osa}

\newcommand{\Gsmall}{H_{\mathrm{s}}}
\newcommand{\Gbig}{H_{\mathrm{b}}}
\newcommand{\epconstant}{\gamma}

\begin{document}

\title{A tight Erd\H{o}s-P\'osa function for wheel minors}

\author[P.~Aboulker]{Pierre Aboulker}
\author[S.~Fiorini]{Samuel Fiorini}
\author[T.~Huynh]{Tony Huynh}
\address[P.~Aboulker, S.~Fiorini, T.~Huynh]{\newline D\'epartement de Math\'ematique
\newline Universit\'e Libre de Bruxelles
\newline Brussels, Belgium}
\email{\{pierreaboulker@gmail.com, sfiorini@ulb.ac.be, tony.bourbaki@gmail.com\}}

\author[G.~Joret]{Gwena\"el Joret}
\address[G.~Joret]{\newline D\'epartement d'Informatique
\newline Universit\'e Libre de Bruxelles
\newline Brussels, Belgium}
\email{gjoret@ulb.ac.be}

\author[J.-F.~Raymond]{Jean-Florent Raymond}
\address[J.-F.~Raymond]{\newline Logic and Semantics Research Group
\newline Technische Universität Berlin
\newline Berlin, Germany}
\email{raymond@tu-berlin.de}

\author[I.~Sau]{Ignasi Sau}
\address[I.~Sau]{\newline CNRS, LIRMM, Universit\'e de Montpellier 
\newline Montpellier, France}
\email{ignasi.sau@lirmm.fr}

\thanks{G.\ Joret is supported by an ARC grant from the Wallonia-Brussels Federation of Belgium. J.-F.\ Raymond is supported by ERC Consolidator Grant 648527-DISTRUCT and has been supported by the Polish National Science Centre grant PRELUDIUM DEC-2013/11/N/ST6/02706. P. Aboulker, S. Fiorini, and T. Huynh are supported by ERC Consolidator Grant 615640-ForEFront.}

\date{\today}
\sloppy

\begin{abstract}
Let $W_t$ denote the wheel on $t+1$ vertices.  We prove that for every integer $t \geq 3$ there is a constant $c=c(t)$ such that for every integer $k\geq 1$ and every graph $G$, either $G$ has $k$ vertex-disjoint subgraphs each containing $W_t$ as a minor, or there is a subset $X$ of at most $c k \log k$ vertices such that $G-X$ has no $W_t$ minor. 
This is best possible, up to the value of $c$. 
We conjecture that the result remains true more generally if we replace $W_t$ with any fixed planar graph $H$. 
\end{abstract}

\maketitle

\section{Introduction}

Let $H$ be a fixed graph. 
An {\em $H$-model} $\mc{M}$ in a graph $G$ is a collection
$\{S_{x} \subseteq G: x\in V(H)\}$ of vertex-disjoint connected subgraphs of $G$
such that $S_{x}$ and $S_{y}$ are linked by an edge in $G$ for every edge $xy \in E(H)$.
The \emph{vertex set} $V(\mc{M})$ of $\mc{M}$ is the union of the vertex sets of the subgraphs in the collection.
Two $H$-models $\mc{M}$ and $\mc{M'}$ are \emph{disjoint} if $V(\mc{M}) \cap V(\mc{M'}) = \varnothing$.

Let $\nu_{H}(G)$ be the maximum number of pairwise disjoint $H$-models in $G$.
Let $\tau_{H}(G)$ be the minimum size of a subset $X \subseteq V(G)$ such that $G - X$
has no $H$-model. 
Clearly, $\nu_{H}(G) \leq \tau_{H}(G)$.  We say that the \emph{\EP{} property holds for $H$-models} if there exists a \emph{bounding function} $f \colon \N \to \mathbb{R}$ such that 
\[
\tau_H (G) \leq f(\nu_H(G)) 
\]
for every graph $G$.

Robertson and Seymour~\cite{robertson1986graph} proved that the \EP{} property holds for $H$-models if and only if $H$ is planar. 
Their original bounding function was exponential.  
However, this has been significantly improved by recent breakthrough results of Chekuri and Chuzhoy \cite{chekuri2013large, chekuri2016polynomial} on the polynomial Grid Theorem. 

\begin{theorem}[Chekuri and Chuzhoy~\cite{chekuri2013large}]
\label{thm:ccep} 
There exist integers $a,b,c \geq 0$ such that for every planar graph $H$ on $h$ vertices, the  \EP{} property holds for $H$-models with bounding function 
\[
f(k) = a h^b \cdot k \log^c (k+1).
\]
\end{theorem}

If we consider $H$ to be fixed and focus solely on the dependence on $k$---which is the point of view we take in this paper---\autoref{thm:ccep} gives a $O(k \log^c k)$ bounding function. 
This is remarkably close to being best possible: 
If $H$ is planar with at least one cycle, then there is a $\Omega(k \log k)$ lower bound on bounding functions. 
This follows easily from the original lower bound of Erd\H{o}s and P\'osa for the case where $H$ is a triangle~\cite{Erdos1965independent}. 
Alternatively, one can see this by considering $n$-vertex graphs $G$ with treewidth $\Omega(n)$ and girth $\Omega(\log n)$ (which exist \cite{morgenstern1994existence}), and notice that $\tau_{H}(G) = \Omega(n)$ (because removing one vertex decreases treewidth by at most one, and $G-X$ has treewidth $O(1)$ when $G-X$ has no $H$-minor, by the Grid Theorem) while $\nu_{H}(G) = O(n / \log n)$ (because each $H$-model contains a cycle). 
Thus, a $O(k \log^c k)$ bound is optimal, up to the value of $c$. 
While no explicit value for $c$ is given in~\cite{chekuri2013large}, a quick glance at the proof suggests that it is at least a double-digit integer. 
In this paper, we put forward the conjecture that a $O(k \log^c k)$ bound holds with $c=1$.

\begin{conjecture}
\label{conj:main}
For every planar graph $H$, the \EP{} property holds for $H$-models with a $O(k \log k)$ bounding function.
\end{conjecture}

If true, \autoref{conj:main} would completely settle the growth rate of the \EP{} functions for $H$-models for all planar graphs $H$ (up to the constant factor depending on $H$).  
That is,  if $H$ is planar with at least one cycle, then the $O(k \log k)$ bound would match the $\Omega(k \log k)$ lower bound mentioned above.  
And if $H$ is a forest, it is already known that the right order of magnitude is $O(k)$, see~\cite{FJW2013}. 

Going back to the $O(k \log^c k)$ bound of Chekuri and Chuzhoy~\cite{chekuri2013large}, one could initially hope that a value of $c=1$ could be obtained by optimizing the various steps of their proof.   
However, any constant $c$ obtained using their general approach necessarily satisfies $c\geq 2$. 
This is because they obtain \autoref{thm:ccep} as a corollary from the following result. 

\begin{theorem}[Chekuri and Chuzhoy~\cite{chekuri2013large}]
\label{thm:cctw} 
There exist integers $a',b',c' \geq 0$ such that for all integers $r, k \geq 1$, every graph $G$ of treewidth at least 
\[
a' r^{b'} \cdot k \log^{c'} (k+1)
\]
has $k$ vertex-disjoint subgraphs $G_1, \dots, G_k$, each of treewidth at least $r$. 
\end{theorem}
Now, if we fix a planar graph $H$ and if $G$ is such that $\nu_{H}(G) < k$, then $G$ cannot have $k$ vertex-disjoint subgraphs each of treewidth at least $r$, where $r=r(H)$ is a constant such that every graph with treewidth at least $r$ contains an $H$ minor. Note that $r(H)$ exists by the Grid Theorem of Robertson and Seymour~\cite{robertson1986graph}.  
Thus, the above theorem implies that $G$ has treewidth $O(k \log^{c'} k)$. 
The authors of~\cite{chekuri2013large} then apply a standard divide-and-conquer approach on an optimal tree decomposition, and obtain a $O(k \log^{c} k)$ bound on $\tau_{H}(G)$ (see~\cite[Lemma~5.4]{chekuri2013large}). 
This unfortunately results in an extra $\log k$ factor;\ $c=c'+1$. 
On the other hand, we must have $c' \geq 1$ in \autoref{thm:cctw}, as shown again by $n$-vertex graphs with treewidth $\Omega(n)$ and girth $\Omega(\log n)$. 
Hence, $c\geq 2$. 
Therefore, one needs a different approach to prove \autoref{conj:main}. 

As a side remark, it is natural to conjecture that we could take $c'=1$ in \autoref{thm:cctw} (at least, if we forget about the precise dependence on $r$): 

\begin{conjecture}
\label{conj:tw}
There is a function $f:\N \to \N$ such that for all integers $r, k \geq 1$, every graph $G$ of treewidth at least 
\[
f(r) \cdot k \log (k+1)
\]
has $k$ vertex-disjoint subgraphs $G_1, \dots, G_k$, each of treewidth at least $r$. 
\end{conjecture}
As it turns out, this conjecture is implied by our \autoref{conj:main}: 
It suffices to take $H$ to be the $r \times r$-grid, which has treewidth $r$. 
Then either $\nu_{H}(G) \geq k$, in which case we are done, or $\nu_{H}(G) < k$, and then there is a subset $X$ of $O(k \log k)$ vertices such that $G-X$ has no $H$-minor, and hence $G-X$ has treewidth at most $g(r)$ for some function $g$ by the Grid Theorem. 
Adding $X$ to all bags of an optimal tree decomposition of $G-X$, we deduce that $G$ has treewidth $O(k \log k)$. 
Thus, this is another motivation to study \autoref{conj:main}. 

While \autoref{conj:main} remains open in general, it is known to hold for some specific graphs $H$.
For example, the original \EP{} theorem~\cite{Erdos1965independent} is simply the assertion that \autoref{conj:main} holds when $H$ is a triangle.  This was recently extended to the case where $H$ is an arbitrary cycle~\cite{FH14} (see also~\cite{BBR07, MNSW17} for related results). The conjecture also holds when $H$ is a multigraph consisting of two vertices linked by a number of parallel edges~\cite{Chatzidimitriou2017}.

Our main result is that \autoref{conj:main} holds when $H$ is a wheel.  A \emph{wheel} is a graph obtained from a cycle by adding a new vertex adjacent to all vertices of the cycle. We denote by $W_t$ the wheel on $t+1$ vertices.

\begin{theorem}
\label{thm:main}
For each integer $t \geq 3$, the \EP{} property holds for $W_t$-models with a $O(k \log k)$ bounding function. 
\end{theorem}

We remark that our main theorem implies all the aforementioned special cases. This is because the existence of a $O(k \log k)$ bounding function for $H$-models is preserved under taking minors of $H$ (see \autoref{lemma:monotone}). Our result also have the following consequence.

\begin{corollary}\label{cor:twbound}
For every $t\in \N$ there is a function $g\colon \N \to \N$ with $g(k) = O(k \log k)$ such that every $(k\cdot W_t)$-minor free graph has treewidth at most $g(k)$.
\end{corollary}

The rest of the paper is organized as follows.  In the next section we present some general lemmas about $H$-models.  Since these lemmas are valid for arbitrary planar graphs $H$, they may be useful in attacking Conjectures~\ref{conj:main} and~\ref{conj:tw}.  In \autoref{sec:proof}, we specialize to the case of wheels and prove our main theorem.  We conclude with some open problems in \autoref{sec:conclusion}.

\section{General Tools} \label{sec:generaltools}

In this paper, our graphs are simple (no parallel edges nor loops).  Let $H,G$ be two graphs.  We let $|G|$ denote $|V(G)|$.  
%
We assume the reader is familiar with the notions of graph minors, tree decompositions, and treewidth (see Diestel~\cite{Diestel} for an introduction to the area). 
We let $\tw(G)$ denote the treewidth of $G$. 

An \emph{$H$-transversal} of $G$ is a set $X$ of vertices of $G$ such that $G - X$ has no $H$-model.  A graph is \emph{minor-minimal} for a given property if it satisfies the property and none of its proper minors does.

We use the following results. The first is an extension of a classic result of Kostochka~\cite{kostochka1984lower} and Thomason~\cite{Thomason1984}, where in addition the size of the $K_t$-model is logarithmic.  
(For definiteness, all logarithms are in base $2$ in this paper.)  

\begin{theorem}[\cite{Montgomery2015loga}, see also \cite{Fiorini20121226, Shapira2015}]\label{thm:small-minors}
There is a function $\varphi(t) = O(t \sqrt{\log t})$ such that, if an $n$-vertex graph has average degree at least $\varphi(t)$, then it contains a $K_t$-model on $O(\log n)$ vertices. 
\end{theorem}

The second is a theorem of Fomin, Lokshtanov, Misra and Saurabh \cite{FominLMS12}, whose original purpose was to show that the algorithmic problem of finding a minimum-size $H$-transversal admits a polynomial-size kernel when $H$ is planar. 

\begin{theorem}[\cite{FominLMS12}]
\label{thm:fominkernel}
For every planar graph $H$, there is a polynomial $\pi$ such that for every $k \in \N$, every graph $G$ with $\tau_H(G) = k$ and minor-minimal with this property satisfies $|G| \leq \pi(k)$.
\end{theorem}

\subsection{Minimal counterexamples to the \texorpdfstring{\EP{}}{Erdös-Pósa} property}

Let $H$ be a graph and let $f \colon \N \to \mathbb{R}$ be a function. We say that a graph $G$ is a \emph{minimal counterexample} to the \EP{} property for $H$-models with bounding function $f$ if the following properties hold:
\begin{enumerate}[(i)]
    \item\label{it:notep} $\tau_H(G) > f(\nu_H(G))$;
    \item\label{it:notep2} subject to the above constraint, $\nu_H(G)$ is minimum;
    \item subject to the above constraints, $|G|$ is minimum;
    \item subject to the above constraints, $|E(G)|$ is minimum.
\end{enumerate}

Notice that the two last requirements of the above definition imply that a minimal counterexample is a minor-minimal graph satisfying requirements \eqref{it:notep} and \eqref{it:notep2}.
The following lemma gives a bound on the size of minimal counterexamples.

\begin{lemma}
\label{lem:min-cex}
Let $H$ be a planar graph and let $f\colon \N \to \R$ be a polynomial non-decreasing function. Then there is a polynomial $\rho$ such that, for every minimal counterexample $G$ to the \EP{} property for $H$-models with bounding function $f$, we have $|G| \leq \rho(\nu_{H}(G))$.
\end{lemma}

\begin{proof}
Let $k := \nu_H(G)$.
Let us first show that $\tau_H(G) = \lfloor f(k)\rfloor +1 $. 
Let $v\in V(G)$. 
Observe that $\tau_H(G) \leq \tau_H(G - v) + 1$.
By minimality of $G$, $\tau_H(G - v) \leq f(\nu_H(G - v))$. 
As $G-v$ is a minor of $G$, we also have $\nu_H(G-v) \leq \nu_H(G)$.
We deduce $\tau_H(G) \leq f(k)+1$. 
Since $G$ is a counterexample, we also have $\tau_H(G) > f(k)$.  
It follows that $\tau_H(G) = \lfloor f(k) \rfloor +1$. 

Now,  $\nu_H(G') \leq \nu_H(G)$ holds for every proper minor $G'$ of $G$, and thus $\tau_H(G') < \tau_H(G)$ (otherwise $G$ would not be a minimal counterexample). 
Hence $G$ is minor-minimal with the property that  $\tau_H(G) = \lfloor f(k) \rfloor +1$.
By \autoref{thm:fominkernel} we obtain $|G| \leq \pi(\lfloor f(k) \rfloor +1)$ where $\pi$ is the polynomial given by that theorem. 
Therefore, it suffices to take $\rho: t \mapsto \pi(\lfloor f(t) \rfloor +1)$. 
\end{proof}

Informally, the following result, originally proved in \cite{FJW2013}, states that if a graph $G$ has a large $H$-minor-free induced subgraph with a small `boundary', then there is a smaller graph $G'$ where the values of $\nu_H$ and $\tau_H$ are the same.

\begin{theorem}[\cite{FJW2013}]\label{fjw13}
For every planar graph $H$, there is a computable function $g'\colon \N\to \N$ such that, for every graph $G$, if $J$ is an $H$-minor-free induced subgraph of $G$ such that exactly $p$ vertices of $J$ have a neighbor in $V(G) \setminus V(J)$ and $|J| \geq g'(p)$, then there exists a graph $G'$ such that $\tau_H(G') = \tau_H(G)$, $\nu_H(G') = \nu_H(G)$, and $|G'| < |G|$.
\end{theorem}

We can use \autoref{fjw13} to upper bound the size of $H$-minor-free induced subgraphs in minimal counterexamples as follows.

\begin{corollary}\label{protrusion}
For every planar graph $H$, there is a computable and non-decreasing function $g\colon \N \to \N$ such that, if $G$ is a minimal counterexample to the \EP{} property for $H$-models with bounding function $f$ for some function $f\colon \N\to\R$, then every $H$-minor free induced subgraph $J$ of $G$ that has exactly $p$ vertices with a neighbor in $V(G) \setminus V(J)$ satisfies $|J|<g(p)$.
\end{corollary}

\begin{proof}
Let $g'$ be the function in \autoref{fjw13}.
We define the function $g\colon \N\to \N$ as follows:
$g(k) = \max_{i \in \{0,\dots,k\}} g'(i)$.
Notice that $g(k) \geq g'(k)$ holds for every $k\in \N$ and that $g$ is non-decreasing. 
Now, suppose that $G$ is a graph having an $H$-minor free induced subgraph $J$ with exactly $p$ vertices having a neighbor in $V(G) \setminus V(J)$ and such that $|J| \geq g(p)$. 
Then, since $g(p) \geq g'(p)$, by \autoref{fjw13} there is a graph $G'$ such that $\tau_H(G') = \tau_H(G)$, $\nu_H(G') = \nu_H(G)$, and $|G'| < |G|$. 
In particular, $G$ cannot be a minimal counterexample to the \EP{} property for $H$-models for any bounding function $f$, a contradiction. 
\end{proof}

\subsection{Interplays between treewidth and the \texorpdfstring{\EP{}}{Erdös-Pósa} property}

Given a planar graph $H$, the standard approach to show that $H$-models satisfy the \EP{} property is to first note that $k\cdot H$ (the disjoint union of $k$ copies of $H$) is also planar. 
Thus, if $\nu_H(G) < k$ for a graph $G$, then the treewidth of $G$ is bounded by a function of $k$ and $H$, by the Grid Theorem~\cite{robertson1986graph}.  
Then one can use a tree decomposition of small width to find a small $H$-transversal of $G$. 
This was first used by Robertson and Seymour~\cite[Theorem~8.8]{robertson1986graph} in their original proof (see also~\cite[Theorem 3]{Thomassen1988} and the survey~\cite[Section~3]{Raymond2017}). 
It was subsequently used by several authors to obtain improved bounding functions, most notably by Chekuri and Chuzhoy~\cite{chekuri2013large} when deriving their \autoref{thm:ccep} from \autoref{thm:cctw}.   

As it was already mentioned in the introduction when discussing \autoref{conj:tw}, the reverse direction holds as well: A bounding function for $H$-models translates directly to an upper bound on the treewidth of $(k\cdot H)$-minor free graphs, up to an additive term depending only on $H$:

\begin{lemma}\label{ceiling} 
Let $H$ be a planar graph, let $f$ be a bounding function for $H$-models, and let $c=c(H)$ be a constant such that $\tw(G) \leq c$ for every $H$-minor free graph $G$. 
Then, for every $k\geq 1$, every $(k\cdot H)$-minor free graph $G$ has treewidth at most~$f(k-1) + c$.
\end{lemma}

\begin{proof} 
Let $G$ be a graph not containing $k\cdot H$ as a minor. 
Since $\nu_H(G) \leq k-1$ and $f$ is a bounding function for $H$-models, we deduce $\tau_H(G) \leq f(k-1)$.
That is, $G$ has a set $X$ of at most $f(k-1)$ vertices such that $G- X$ is $H$-minor free.
By definition of $c$, we have $\tw(G - X) \leq c$. Then $\tw(G) \leq c + |X| \leq c + f(k-1)$, as desired.
\end{proof}

Thus combining our main result with \autoref{ceiling} gives \autoref{cor:twbound} (stated in the introduction).

We also include the following lemma, which states that if $H'$ is a minor of $H$, then a bounding function for $H'$-models can be easily obtained from a bounding function for $H$-models. 

\begin{lemma}\label{lemma:monotone} 
Let $H$ be a fixed planar graph, let $f$ be a bounding function for $H$-models, and let $c=c(H)$ be a constant such that $\tw(G) \leq c$ for all $H$-minor free graphs $G$.
If $H'$ is a minor of $H$ with $q$ connected components, then $k \mapsto f(k) + (qk-1)(c+1)$ is a $O(f(k))$ bounding function for $H'$-models.
\end{lemma}

\begin{proof}
Let $G$ be a graph with $\nu_{H'}(G) \leq k$. As $H'$ is a minor of $H$, we deduce $\nu_H(G) \leq k$. By definition of $f$, 
there is a set $X$ of at most $f(k)$ vertices such that $G- X$ is $H$-minor free. Hence $\tw(G-X) \leq c$.
Theorem 8.8 in~\cite{robertson1986graph} provides the following upper-bound on $\tau$ in graphs of bounded treewidth: 
\[\tau_{H'}(G - X) \leq (qk-1)(c+1).\]
Then, $\tau_{H'}(G) \leq \tau_{H'}(G - X) + |X| \leq (qk-1)(c+1) + f(k)$. 
Finally, since $f(k) \geq k$, we deduce $(qk-1)(c+1) + f(k) = O(f(k))$.
\end{proof}

\section{The Proof for Wheels} \label{sec:proof}

In this section, we prove our main theorem:

{
\renewcommand{\thetheorem}{\ref*{thm:main}}
\begin{theorem}
For each integer $t \geq 3$, the \EP{} property holds for $W_t$-models with a $O(k \log k)$ bounding function.
\end{theorem}
\addtocounter{theorem}{-1}
}
    
\begin{proof}
To keep track of the dependencies between the constants that we use, we define them here. Recall that $t$ denotes the number of spokes of the wheel that we are considering.
Let $\varphi$ and $\varphi'$ be constants such that every $n$-vertex graph of average degree at least $\varphi$ has a $K_{t+1}$-model on at most $\varphi' \log n$ vertices (both $\varphi$ and $\varphi'$ depend on $t$, see \autoref{thm:small-minors}).
Let $\alpha,\beta\geq 1$ be constants such that $\rho(n) \leq \alpha n^{\beta}$, for every $n \in \N\setminus \{0\}$, where $\rho$ is the polynomial of \autoref{lem:min-cex} for $H=W_t$.  
Let $g$ denote the function from \autoref{protrusion} for $H=W_t$. 

We then set
\begin{align*}
c_1 &=  g(2t\varphi^2), \quad p = g(2c_1\varphi^2), \quad c_2 = 4p,\\
\sigma &= \max\left\{3 \varphi' c_2, 2c_2 + tp, (2t^2p+1)(c_2+2c_1\varphi^2)\right\},\quad \text{and}\\
\epconstant{} &= \sigma(\beta  + \log \alpha).
\end{align*}

Observe that we have $t < c_1 < p < c_2$.
Let $f(k) := \epconstant{}  \cdot k\log(k+1)$, for every $k \in \N$. We show that the  \EP{} property holds for $W_t$-models with bounding function~$f$.

Arguing by contradiction, let $G$ be a minimal counterexample to the \EP{} property for $W_t$-models with bounding function~$f$. 
Let $k := \nu_{W_t}(G)$ Then $k\geq 1$ and $|G|$ is polynomial in $k$ by \autoref{lem:min-cex}. That is, letting $n := |G|$, we have $n \leq \alpha k^{\beta }$, for the constants $\alpha$ and $\beta $ defined above.

We first show that $G$ cannot contain a $W_t$-model of logarithmic size.
\begin{claim}\label{cl:small-model}
$G$ has no $W_t$-model of size at most~$\sigma \log n$. 
\end{claim}
\begin{proof}
Towards a contradiction, we consider a $W_t$-model $\mathcal{M}$ of size at most $\sigma \log n$. 
Notice that $\log n \leq (\beta  + \log \alpha)\log (k+1)$ can be deduced from the aforementioned upper-bound on $n$.
Since $\nu_{W_t}(G - V(\mathcal{M})) \le k-1$, by minimality of $G$, 
\begin{align*}\tau_{W_t}(G) &\le |V(\mathcal{M})| + \tau_{W_t}(G - V(\mathcal{M}))\\
              &\le \sigma \log n + f(k-1)\\
              &\le \epconstant{} \log (k+1) + \epconstant{} (k-1) \log k\\
              &\le \epconstant{} k \log (k+1) \leq f(k).
\end{align*}
However, this contradicts the fact that $G$ is a minimal counterexample to the \EP{} property for $W_t$-models with bounding function~$f$. 
\end{proof}
Note that since $\nu_{W_t}(G) \geq 1$, we have $n\geq t+1\geq 2$, and thus $\log n \geq 1$ (recall that all logarithms are in base $2$). Thus, \autoref{cl:small-model} implies in particular that $G$ has no $W_t$-model of size at most $\sigma$.

Let $\mc C$ be a maximum-size collection of vertex-disjoint cycles in $G$ whose lengths are in the interval $[c_1,c_2]$. Let $\mc P$ be a maximum-size collection of vertex-disjoint paths of length $p$ in $G - V(\mc C)$, 
where $V(\mc C) := \bigcup_{C \in \mc C} V(C)$. 
(In this paper, the length of a path is defined as its number of edges.) 
Finally, let $\mc R$ be the collection of   components of $G - (V(\mc C) \cup V(\mc P))$ and let $V(\mc R) := V(G) - (V(\mc C) \cup V(\mc P))$, where $V(\mc P) := \bigcup_{P \in \mc P} V(P)$. We point out that the cycles in $\mc{C}$ and the paths in $\mc{P}$ are subgraphs of $G$ but not necessarily induced subgraphs of $G$, while the components in $\mc{R}$ are induced subgraphs of $G$. 
We call the elements of $\mc{C} \cup \mc{P} \cup \mc{R}$ \emph{pieces}.

Observe that, by maximality of $\mc P$, every path in a piece of  $\mc R$ has length at most $p-1$. This implies that each such piece is $W_t$-minor free. 
Indeed, observe that if such a piece $R$ has a $W_t$-model then $R$ contains a subgraph consisting of a cycle $C$ and a rooted tree $T$ such that $T$ has at most $t$ leaves, $V(T) \cap V(C)=\emptyset$, and the leaves of $T$ collectively have at least $t$ neighbours in $C$. 
The cycle $C$ has at most $p$ vertices, and each root-to-leaf path in $T$ has at most $p$ vertices. 
Thus, this gives a $W_t$-model with at most $(t+1)p$ vertices. 
However, this contradicts \autoref{cl:small-model} since $(t+1)p \le \sigma $. 

Similarly, each piece in  $\mc C$ (in $\mc P$, respectively) has at most $c_2$ ($p$, respectively) vertices, and these vertices induce a $W_t$-minor free subgraph of~$G$; otherwise, there would exist a $W_t$-model of size at most $c_2$ (resp.~$p$), again a contradiction to \autoref{cl:small-model} since $p \leq c_2 \leq \sigma$. 
These  facts will be used often in the rest of the proof. 

We say that two distinct pieces $K$ and $K'$ \emph{touch} if some edge of $G$ links some vertex of $K$ to some vertex of $K'$. Note that, by construction, two distinct pieces in $\mc{R}$ cannot touch.  A piece is said to be \emph{central} if it is a cycle in $\mc{C}$, a path in $\mc{P}$, or a piece in $\mc{R}$ that touches at least $2\varphi$ other pieces.
In the next paragraph, we define two auxiliary simple graphs $\Gsmall$ (for small degrees) and $\Gbig$ (for big degrees) that model how the central pieces are connected through the noncentral pieces. To keep track of the correspondence between the edges of $\Gsmall$ and the noncentral pieces, we put labels on some of these edges.

Initialize both $\Gsmall$ and $\Gbig$ to the graph whose set of vertices is the set of central pieces and whose set of edges is empty. For each pair of central pieces that touch in $G$, add an (unlabeled) edge between the corresponding vertices in both $\Gsmall$ and $\Gbig$.

Next, while there is some noncentral piece $R \in \mc R$ that touches two central pieces $K$ and $K'$ that are not yet adjacent in $\Gbig$, call $\mathcal{Z}_R$ the set of central pieces that touch $R$ and do the following:
\begin{enumerate}
    \item \label{enum:gb} add all (unlabeled) edges to $\Gbig$ between pieces of $\mathcal{Z}_R$ (not already present in $\Gbig$). This creates a clique on vertex set $\mathcal{Z}_R$ in $\Gbig$, some of whose edges might have already been there before.
    \item \label{enum:gs} choose a piece $K\in \mathcal{Z}_R$ such that the number of newly added edges of $\Gbig$ incident to $K$ is maximum. Add to $\Gsmall$ every edge that links $K$ to another piece of $\mathcal{Z}_R$ (not already present in $\Gsmall$), and label it with $R$. This creates a star centered at $K$ in $\Gsmall$ with all its edges labeled with the noncentral piece $R$.
\end{enumerate}

\begin{figure}[h]
    \centering
    \begin{tikzpicture}[every node/.style = black node]
    \begin{scope}
        \draw
            (0,0) node (a) {}
            (1,0) node (b) {}
            (1,1) node (c) {}
            (0,1) node (d) {};
            \draw (a) -- (b) -- (c) -- (d);
        \draw (-1.25, 0.5) node[normal] {(a)};
    \end{scope}
    \begin{scope}[xshift = 4cm]
        \draw
            (0,0) node (a) {}
            (1,0) node (b) {}
            (1,1) node (c) {}
            (0,1) node (d) {};
            \draw (a) -- (b) -- (c) -- (d);
            \draw[very thick, color = red!80!black]
                (d) -- (a)
                (a) -- (c)
                (d) -- (b);
    \draw (-1.25, 0.5) node[normal] {(b)};
    \end{scope}
    \begin{scope}[xshift = 8cm]
        \draw
            (0,0) node (a) {}
            (1,0) node (b) {}
            (1,1) node (c) {}
            (0,1) node[label = 180:$K$] (d) {};
            \draw (a) -- (b) -- (c) -- (d);
            \draw[very thick, color = red!80!black]
                (d) -- (a) node[normal, midway, label = 180:$R$] {}
                (d) -- (b) node[normal, midway, label = 45:$R$] {};
        \draw (-1.25, 0.5) node[normal] {(c)};
    \end{scope}
    \end{tikzpicture}
    \caption{Construction of $\Gsmall$ and $\Gbig$. (a):~the vertices of a set $\mathcal{Z}_R$ in $\Gsmall$ and $\Gbig$ before step~\eqref{enum:gb}. (b):~$\Gbig[\mathcal{Z}_R]$ after step~\eqref{enum:gb}. (c):~$\Gsmall[\mathcal{Z}_R]$ after step~\eqref{enum:gs}.}
    \label{fig:constr}
\end{figure}
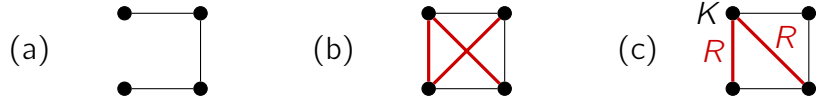

The edges added during these steps are depicted in thick red lines in the example of~\autoref{fig:constr}.
By construction, $\Gsmall$ is a subgraph of $\Gbig$. These graphs have the following two crucial properties.

\begin{claim}\label{claim:model}
If $\Gsmall$ has a $W_t$-model of size $q$, then $G$ has a  $W_t$-model of size at most $3c_2q$. 
\end{claim}
\begin{proof}
\newcommand{\Msmall}{M_{\mathrm{s}}}
Suppose that $\Gsmall$ has a $W_t$-model of size $q$. 
Then there exists a subgraph $\Msmall \subseteq \Gsmall$ with $q$ vertices that can be contracted to $W_t$. 
We may assume that the average degree of $\Msmall$ is at most that of $W_t$, and hence at most~$4$. 
That is, $|E(\Msmall)| \le 2 |\Msmall|$. From the subgraph $\Msmall$, we construct a subgraph $M \subseteq G$ that can be contracted to $\Msmall$, and thus also to $W_t$. 

First, for each central piece $K \in V(\Msmall) \cap (\mc{C} \cup \mc{P})$, we add all its vertices to $M$, as well as $|K|-1$ edges from $K$ in such a way that the subgraph of $M$ induced by $V(K)$ is connected. For each central piece $K \in V(\Msmall) \cap \mc{R}$, we choose some vertex $v_K \in K$ and add it to $M$. This creates at most $c_2|\Msmall| = c_2 q$ vertices in $M$. 

Second, for each unlabeled edge $KK'$ of $\Msmall$ with $K, K' \in V(\Msmall) \cap (\mc{C} \cup \mc{P})$, we choose some edge of $G$ linking $K$ to $K'$ and add it to $M$. This does not create any new vertex in $M$.

Third, for each edge $KK'$ of $\Msmall$ that has not been considered so far, we add to $M$ a path linking some vertex of $V(M) \cap V(K)$ to some vertex of $V(M) \cap V(K')$, as follows. If the edge $KK'$ is not labeled, then exactly one of its endpoints is a central piece in $\mc{R}$, say $K$.  
The path we add to $M$ links $v_K$ to some vertex of $K'$ and is a subgraph of $K$, except for the last edge and last vertex. Thus, this path has at most $p - 1$ internal vertices.  
If the edge $KK'$ is labeled with the noncentral piece $R \in \mc{R}$, then this edge is part of a star in $\Msmall$ whose edges are all labeled with $R$. We may assume without loss of generality that $K$ is the center of this star. 
In this case, the path we add to $M$ links some vertex of $K$ to some vertex of $K'$ and has all its internal vertices in $R$. 
Thus, this path has at most $p$ internal vertices. 

In total, the addition of these paths to $M$ creates at most $p |E(\Msmall)| \leq 2 c_2 |\Msmall| = 2 c_2 q$ new vertices in $M$. 
The resulting subgraph $M$ has at most $c_2 |\Msmall| + p |E(\Msmall)| \le 3 c_2 q$ vertices. 
By construction, $M$ can be contracted to $\Msmall$, as desired. 
\end{proof}

\begin{claim}\label{claim:av}
The average degree of $\Gbig$ is at most $\varphi$ times the average degree of $\Gsmall$. The degree of each central piece of $\mc{R}$ in $\Gsmall$ is at least $2\varphi$.
\end{claim}
\begin{proof} 
First, note that edges that appear in $\Gbig$ but not in $\Gsmall$ must be labeled. 
Let $R \in \mc R$ be a noncentral piece, and let $r$ be the number of pieces in $\mc C \cup \mc P$ it touches. By definition of noncentral pieces,  $r <2\varphi$. When $R$ is treated in the algorithm used to construct $\Gbig$ and $\Gsmall$, if $q$ new edges are added to $\Gbig$, then one of the pieces touched by $R$ is incident to at least $2q/r>q / \varphi$ of these new edges and thus at least $q / \varphi$ new edges are added to $\Gsmall$. This proves the first part of the claim. 

By definition, a piece $K$ of $\mc{R}$ is central if it touches at least $2\varphi$ other pieces. As two pieces of $\mc{R}$ cannot touch, $K$ touches at least $2\varphi$ pieces from $\mc C \cup \mc P$, that is, at least $2 \varphi$ other central pieces. Then in the first step of the construction of $\Gsmall$, all edges have been added from $K$ to these pieces.
\end{proof}

If the average degree of $\Gsmall$ is at least $\varphi$, then by definition of $\varphi$ and $\varphi'$ at the beginning of the proof, $\Gsmall$ has a $K_{t+1}$-model of size at most $\varphi'\log |\Gsmall|$, and thus in particular a $W_t$-model of size at most $\varphi'\log |\Gsmall|$. By \autoref{claim:model}, this gives a $W_t$-model of size at most $3\varphi'c_2 \log |\Gsmall| \leq 3\varphi' c_2 \log n$ in $G$, a contradiction to \autoref{cl:small-model} since $3\varphi' c_2 \leq \sigma$.

Thus, the average degree of $\Gsmall$ is smaller than $\varphi$. 
Hence, by \autoref{claim:av}, the average degree of $\Gbig$ is  smaller than $\varphi^2$. 
Then strictly more than half of the central pieces have degree less than $2\varphi$ in $\Gsmall$ (otherwise at least half of the vertices of $\Gsmall$ have degree at least $2\varphi$, a contradiction to the fact that $\Gsmall$ has average degree less than $\varphi$). Similarly, strictly more than half of the central pieces have degree less than $2\varphi^2$ in $\Gbig$. 
Thus there is a central piece whose degree in $\Gsmall$ is less than $2\varphi$, and whose degree in $\Gbig$ is less than $2\varphi^2$. Choose such a piece $K$. By \autoref{claim:av} (second part of the statement), $K$ is either   in $\mc{C}$ or   in $\mc{P}$.

In the rest of the proof we use the fact that $K$ has degree less than $2\varphi^2$ in $\Gbig$ 
to find a $W_t$-model of size at most $\sigma \log n$, contradicting \autoref{cl:small-model}. 

If $K$ is the unique central piece in $\mc C \cup \mc P$, then $V(K)$ is a $W_t$-transversal of $G$ since each piece in $\mc R$ is $W_t$-minor free. 
Thus $\tau_{W_t}(G) \leq |K| \leq c_2 \leq f(k)$, contradicting the fact that $G$ is a counterexample. 

For each central piece $K'$ adjacent to $K$ in $\Gbig$, we consider the collection $\mc R_{K,K'}$ of all noncentral pieces $R \in \mc{R}$ that touch both $K$ and $K'$ ($\mc R_{K,K'}$ might be empty). Then we consider the subgraph $G_{K'}$ of $G$ induced by $V(K) \cup V(K') \cup V(\mc{R}_{K,K'})$. 

Let $q$ be an integer equal to $t$ if $K \in \mc C$, to $c_1$ if $K \in \mc P$.

Our next goal is to show that for every central piece $K'$ adjacent to $K$ in $\Gbig$, there exists a set of  strictly less than $q$  vertices that separates $K$ from $K'$ in $G_{K'}$. 
Thus fix a piece  $K'$ adjacent to $K$ in $\Gbig$. 
By Menger's theorem, it suffices to show  that the maximum number of vertex-disjoint $K$--$K'$ paths in $G_{K'}$ is strictly less than $q$. 
Assume  for contradiction that $G_{K'}$ contains $q$ vertex-disjoint $K$--$K'$ paths. 

By taking the paths to be as short as possible, we may assume that only their endpoints are in $K$ and $K'$, all their internal vertices are in pieces in $\mc R_{K,K'}$, and each such path intersects at most one piece in $\mc R_{K,K'}$ and thus has length at most $p+1$. 

Assume first that $K \in \mc C$, and so $q=t$.
In this case $G_{K'}$ contains a small $W_t$-model as follows.  Let $T$ be a smallest tree in $K'$ containing all the endpoints of our paths in $K'$.
The center vertex of the wheel is then modeled by the union of $T$ and the $t$ $K$--$K'$ paths (minus their endpoints in $K$). 
If $K' \in \mc C \cup \mc P$, then obviously $|V(T)| \leq c_2$, and the model  thus has at most $2c_2+tp$ vertices. 
If $K' \in \mc R$, then $|V(T)| \leq tp$ since each path in $K'$ has length at most $p-1$; moreover, $\mc R_{K,K'}$ is empty in this case, implying that the model has at most $c_2+tp$ vertices. 
Therefore, in both cases the resulting model has at most $2c_2+tp$ vertices, which contradicts  \autoref{cl:small-model} since  $2c_2+tp \leq \sigma$.
 
Assume now that $K \in \mc P$.    
Since $t<c_1$, by the previous case we may assume that $K' \in \mc P \cup \mc R$.  Since there are $q=c_1$ vertex-disjoint $K$--$K'$ paths in $G_{K'}$, two of these paths intersect $K$ on two vertices that are at distance at least $c_1-1$ on the path $K$, which allows us to construct a cycle in $G_{K'}$ of length at least $c_1$  and at most $4p$: The cycle might use all the vertices of $K$ and at most $p$ vertices of $K'$, which is at most $2p$ vertices, and might intersect at most two pieces of $\mc R_{K,K'}$, using at most $p$ vertices in each of them. 
This is a contradiction to the maximality of $\mc C$: The length of this cycle is in the interval $[c_1, c_2]$ and yet the cycle is vertex disjoint from all cycles in $\mc C$.

Therefore, for each $K'$ adjacent to $K$ in $\Gbig$,  there exists a set $X(K')$ with less than $q$  vertices meeting all the $K$--$K'$ paths in $G_{K'}$. 

Let $X := \bigcup_{K'} X(K')$ where the union is taken over all central pieces $K'$  adjacent to $K$ in $\Gbig$. Note that $|X| \le 2 q \varphi^2$ since there are at most $2\varphi^2$   such pieces $K'$ and for every $K'$ we have  $|X(K')| \le q$. 

We also note that $X$ separates $K$ from all other central pieces in $G$. To see this, let $K''$ be a central piece distinct from $K$ and let $Q$ be a
$K''$--$K$ path in $G$.  Let $K'$ be the last central piece that $Q$ meets before reaching $K$. Then $Q$ contains a $K'$--$K$ path that is contained in $G_{K'}$, which must contain a vertex from $X$.

Let $J$ be the union of the components of $G - X$ that intersect $K$.
Observe that $V(K)$ is not completely included in $X$: If $K \in \mc C$, then $|K| \ge c_1 > 2t\varphi^2 \ge |X|$, and if $K \in \mc P$, then $|K| = p > 2c_1\varphi^2 \ge |X|$. Thus $J$ is not empty. Note also that $X$ separates $J$ from the rest of the graph.

Suppose that the subgraph $G'$ of $G$ induced by $X \cup V(J)$ is $W_t$-minor free. Thus, by \autoref{protrusion}, $|G'| < g(|X|)$. We deduce
\begin{align*}
|X| + |J| &< g(|X|) &\text{since}\ |G'| = |X| + |J|\\
|K| &< g(|X|) &\text{since}\ |J| \ge |K|-|X|\\
|K|    &< g(2q\varphi^2) &\text{since $g$ is non-decreasing}.
\end{align*}

Hence, if $K \in \mc C$, then $c_1 \leq |K| < g(2t\varphi^2)$,
and if $K \in \mc P$, then $p \leq |K| < g(2c_1\varphi^2)$. 
Since $c_1=g(2t\varphi^2)$ and $p=g(2c_1\varphi^2)$, we get a contradiction in both cases. 

Thus, we may assume that $G'$ contains a $W_t$-model. 
Let $M$ be a subgraph of $G'$ containing $W_t$ as a minor with $|V(M)| + |E(M)|$ minimum. 
(We remark that here we take $M$ to be a subgraph instead of just a model as before because we will need to consider the edges of that subgraph in the proof.) 
To finish the proof, it is now enough to prove that $M$ has  at most $\sigma\log n$ vertices, since by \autoref{cl:small-model} this will give us the desired contradiction.

Let $R(J):=J[V(\mc R)]$. Thus $R(J)$ consists of a number of disjoint pieces or subgraphs of pieces of $\mc R$.
Note that $M$ might use all vertices of $V(K) \cup X$ (which is fine); what we need to prove is that it cannot use too many vertices of $R(J)$.  

First, suppose that $M$ is fully contained in some piece of $\mc R$.
Since the vertices of $M$ can be covered with $2t$ paths, and each path in the piece has length less than $p$, it follows that $|M| \leq 2tp \leq \sigma$ and we are done. 

Thus we may assume that $M$ is not contained in some piece of $\mc R$, and thus in particular $M$ is not contained in $R(J)$ (since $M$ is connected).  
By the above remark, we also know that each component of $M[V(R(J))]$ contains at most $2tp$ vertices.
Since $M$ has maximum degree at most $t$ (by minimality of $|V(M)| + |E(M)|$), there are at most $t|V(K) \cup X|$ edges of $M$ with one endpoint in $V(K)\cup X$ and the other in $R(J)$. Hence $M$ intersects $R(J)$ on at most $2t^2p|V(K) \cup X|$ vertices. 
Therefore, $M$ has at most $2t^2p|V(K) \cup X| + |V(K) \cup X|$ vertices. 
Since $|V(K) \cup X| \le |K|+|X| \le c_2 + 2c_1 \varphi^2$, we deduce that $|M| \le (2t^2p+1)(c_2 + 2c_1 \varphi^2)\leq \sigma$, as desired. 
\end{proof}

+\section{Conclusion} \label{sec:conclusion}

One obvious extension of our result for wheels would be to prove it for all planar graphs.
Note that the first steps of our proof work for any such $H$: 
Starting with $G$ a minimal counterexample for some bounding function and some value $k$, we have that $G$ has $n \leq k^{O(1)}$ vertices. 
Thus, in order to get a contradiction, it is enough to show that there is a $O(\log n)$-size $H$-model in $G$. 
Unfortunately, the rest of our proof is specific to wheels and does not generalize. 

Let us mention another possible extension of our result. 
Strengthening the $O(k \log k)$ bound from~\cite{FH14}, Mousset, Noever, \v Skori\'c, and Weissenberger~\cite{MNSW17} recently showed that there is a constant $c > 0$ such that for every $\ell \geq 3$, models of the $\ell$-cycle $C_{\ell}$ have the \EP{} property with bounding function $c k\log k + \ell k$. 
In particular, the constant $c$ in front of the $k\log k$ term is independent of $\ell$. 
We expect that a similar property holds for wheels: 

\begin{conjecture}
There is a constant $c>0$ and a function $g:\N \to \N$ such that for all integers $t \geq 3$, $W_t$-models have the \EP{} property with bounding function 
\[
c k \log k + g(t) k.
\]
\end{conjecture}

\textbf{Acknowledgements.} We thank Stefan van Zwam for valuable discussions at an early stage of this project. We also thank Hong Liu for spotting a slight inaccuracy in the proof of \autoref{thm:main} in an earlier version of this paper.  Finally, we thank the referees for their careful reading and many helpful comments.


\end{document}